\documentclass{article}
\usepackage{amsmath}
\usepackage{amssymb}
\usepackage{amsfonts,,mathdots}
\usepackage{latexsym,longtable}
\usepackage[colorlinks,linkcolor=blue,anchorcolor=blue,citecolor=green,CJKbookmarks=True]{hyperref}
\usepackage{tikz}
\usetikzlibrary{matrix}
\begin{document}
\newcommand{\B}{{\cal B}}
\newcommand{\D}{{\cal D}}
\newcommand{\E}{{\cal E}}
\newcommand{\F}{{\cal F}}
\newcommand{\A}{{\cal A}}
\newcommand{\Hh}{{\cal H}}
\newcommand{\Pp}{{\cal P}}
\newcommand{\Z}{{\bf Z}}
\newcommand{\T}{{\cal T}}
\newcommand{\ZZ}{{\mathbb{Z}}}
\newcommand{\qed}{\hphantom{.}\hfill $\Box$\medbreak}
\newcommand{\proof}{\noindent{\bf Proof \ }}
\renewcommand{\theequation}{\thesection.\arabic{equation}}
\newtheorem{theorem}{Theorem}[section]
\newtheorem{lemma}[theorem]{Lemma}
\newtheorem{corollary}[theorem]{Corollary}
\newtheorem{remark}[theorem]{Remark}
\newtheorem{example}[theorem]{Example}
\newtheorem{definition}[theorem]{Definition}
\newtheorem{construction}[theorem]{Construction}
\newtheorem{open}[theorem]{Open problem}
\newtheorem{claim}[theorem]{Claim}
\def\f#1{{\mathbb{F}}_{#1}}


\medskip
\title{Linear $(2,p,p)$-AONTs do Exist\thanks{Research supported by NSFC grant 11431003 (L. Ji). E-mail: xinw@suda.edu.cn, cuijie@suda.edu.cn, jilijun@suda.edu.cn
 }}

 \author{{\small    Xin Wang\thanks{Corresponding author}, Jie Cui and Lijun Ji} \\
 {\small Department of Mathematics, Soochow University, Suzhou 215006, China}
}

\date{}
\maketitle
\begin{abstract}
\noindent \\ A $(t,s,v)$-all-or-nothing transform (AONT) is a bijective mapping defined on $s$-tuples over an alphabet of size $v$, which satisfies that if any $s-t$ of the $s$ outputs are given, then the values of any $t$ inputs are completely undetermined. When $t$ and $v$ are fixed, to determine the maximum integer $s$ such that a $(t,s,v)$-AONT exists is the main research objective. In this paper, we solve three open problems proposed in [IEEE Trans. Inform. Theory 64 (2018), 3136-3143.] and show that there do exist linear $(2,p,p)$-AONTs. Then for the size of the alphabet being a prime power, we give the first infinite class of linear AONTs which is better than the linear AONTs defined by Cauchy matrices. Besides, we also present a recursive construction for general AONTs and a new relationship between AONTs and orthogonal arrays.

\medskip

\medskip


\end{abstract}


\section{Introduction}

The investigation of all-or-nothing transforms dates back to \cite{Rivest}, in which Rivest suggested using it as a preprocessing for block ciphers in the setting of computational security. However, little attention was attracted to this topic until Stinson proposed unconditionally secure all-or-nothing transforms in \cite{Stinson2001}. Later D'Arco et al. \cite{DES2016} introduced more general types of unconditionally secure all-or-nothing transforms.

We begin with the following definition.

\begin{definition}
Let $X$ be a finite set known as an alphabet. Let $s$ be a positive integer and consider a map $\phi: X^s \rightarrow X^s$. For an input $s$-tuple, say $x=(x_1,\dots,x_s)$, $\phi$ maps it to an output $s$-tuple, say $y=(y_1,\dots,y_s)$, where $x_i,y_i\in X$ for $1\le i \le s$. The map $\phi$ is an unconditionally secure {\it t-all-or-nothing transform} provided that the following properties are satisfied:

$\bullet$ $\phi$ is a bijection.

$\bullet$ If any $s-t$ out of the $s$ output values $y_1,\dots,y_s$ are fixed, then any $t$ of the input values $x_i$ $(1\le i \le s)$ are completely undetermined, in an information-theoretic sense.
\end{definition}

We will call such a map $\phi$ as a $(t,s,v)$-AONT, where $v=|X|$. And when $s$ and $v$ are clear or not relevant, we just call it a $t$-AONT.

The theory of AONTs is at a rudimentary stage. The study of Rivest \cite{Rivest} and Stinson \cite{Stinson2001} concentrated on the case $t=1$. The $1$-AONTs can provide a preprocessing called ``package transform'' for block ciphers. The idea is to use a $1$-AONT to encrypt plaintexts $(x_1,\dots,x_s)$ to $(y_1,\dots,y_s)=\phi(x_1,\dots,x_s)$. Due to the property of $1$-AONT, a partial decryption cannot provide any information about each symbol among the plaintexts. In \cite{DES2016}, the authors mainly concerned the case $t=v=2$ and introduced ``approximations'' to AONT. In this case, more theoretical results could be found in \cite{Zhang} and additional computational results could be found in \cite{NS2017}. Recently, Nasr Esfahani et al. \cite{Stinson2017} concentrated on the case $t=2$ over arbitrary alphabets and suggested $8$ interesting open problems.

 An AONT with alphabet $\mathbb{F}_q$ is linear if each $y_i$, $i\in \{1,\ldots,s\}$ is an
$\mathbb{F}_q$-linear function of $x_1, \ldots, x_s$. Then, we can write
\begin{equation}
\label{def-LAONT}
(y_{1},y_{2},\ldots,y_{s})= (x_1,\ldots,x_s)M^{-1}\ {\rm and}\ (x_1,\ldots,x_s) = (y_{1},y_{2},\ldots,y_{s})M,
\end{equation}
where $M$ is an invertible $s$ by $s$ matrix with entries from  $\mathbb{F}_q$. Subsequently, when we refer
to a ``linear AONT", we mean the matrix $M$ that transforms $(y_{1},y_{2},\ldots,y_{s})$ to $(x_1,\ldots,x_s)$, as specified in (\ref{def-LAONT}).

D'Arco et al. characterized linear all-or-nothing transforms in terms of
submatrices of the matrix $M$ as follows.

\begin{lemma}
 [{\rm \cite[Lemma 1]{DES2016}}] Suppose that $q$ is a prime power and $M$ is an invertible $s$ by $s$ matrix with entries from $\mathbb{F}_q$. Then $M$ defines a linear $(t,s,q)$-AONT if and only if every $t$ by $t$ submatrix of $M$ is invertible.
 \end{lemma}

Next, we review some known results on linear AONTs.

\begin{theorem}[\rm{\cite[Theorem 2]{Stinson2001}}]\label{cauchy}
Suppose $q$ is a prime power and $q\ge 2s$. Then there is a linear transform that is simultaneously a $(t,s,q)$-AONT for all $t$ such that $1\leq t\leq s$.
\end{theorem}

\begin{theorem}[\rm{\cite[Theorem 14]{Stinson2017}}]\label{upper}
There is no linear $(2,q+1,q)$-AONT for any prime power $q>2$.
\end{theorem}

Given a prime power $q$, define
$$S(q)=\{s:\textup{there exists a linear $(2,s,q)$-AONT}\}.$$

By Theorems \ref{cauchy} and \ref{upper}, $S(q)$ is well defined and $S(q)\ne\emptyset$, the maximum element in $S(q)$ is denoted by $M(q)$.

 In this paper, we continue the study of $2$-AONTs. Our main contributions are as follows:

\begin{itemize}
  \item We give a negative answer to the open problem (4) in \cite{Stinson2017}.

  {\bf \cite[Open Problem 4]{Stinson2017}:} As mentioned in Section 2.2, we performed exhaustive searches for linear $(2,q,q)$-AONT in type $q-1$ standard form, for all primes and prime powers $q \leq 9$, and found that no such AONT exists. We ask if there exists any linear $(2,q,q)$-AONT in type $q-1$ standard form.
  \item  By establishing a connection between linear AONTs constructed and cyclic codes, we give positive answers to the open problems (1) and (2) in \cite{Stinson2017}.

     {\bf \cite[Open Problem 1]{Stinson2017}:} Are there infinitely many primes $p$ for which there exist linear $(2,p,p)$-AONT?

      {\bf \cite[Open Problem 2]{Stinson2017}:} Are there infinitely many primes $p$ for which there exist cyclic skew-symmetric $(2,p,p)$-AONT?

      As a consequence, when $p$ is a prime, $M(p)$ is completely determined.
  \item When $q$ is a prime power, we construct a $(2,\Phi(q),q)$-AONT and improve the lower bound on $M(q)$ in general. To the best of our knowledge, this is the first infinite class of AONTs  which is better than the AONTs defined by Cauchy matrices.
  \item For general AONTs, we present a recursive construction for general AONTs and a new relationship between AONTs and orthogonal arrays, and get a general construction for nonlinear $(2,3,n)$-AONT, except for $n=2,6$.
\end{itemize}

We summarize upper and lower bounds on $M(q)$ in Table \ref{table} and the main contributions in this paper are in bold form.

\begin{table}[!t]
\centering
\caption{Upper and Lower bounds on $M(q)$}
\label{table}
\begin{tabular}{c|c}\hline
  bound & authority \\
  \hline
  \textbf{$M(p)=p$ for $p$ is prime} & \textbf{Theorem \ref{main}} and Theorem 14 \cite{Stinson2017}  \\\hline
  \textbf{$\Phi(q)\leq M(q)\leq q$ for all prime powers $q$} & \textbf{Theorem \ref{main2}} and Theorem 14 \cite{Stinson2017} \\\hline
  $M(q)\geq q-1$ if $q=2^n$ and $q-1$ is prime & Theorem 11 \cite{Stinson2017} \\\hline
  $M(4)=4$ & Example 29 \cite{Stinson2017}\\
  $M(8)=7$ & Theorem 11 \cite{Stinson2017}\\
  $M(9)=8$ & Example 30 \cite{Stinson2017} \\

\end{tabular}
\end{table}

The rest of this article is organized as follows. Section~\ref{2} concerns linear $2$-AONT and the open problems proposed in \cite{Stinson2017}. Section~\ref{3} shows a new relation between general AONTs and orthogonal arrays and a construction for nonlinear $(2,3,n)$-AONT. A conclusion is made in Section~\ref{4}.

\section{Linear AONT}\label{2}

In this section, we describe some theoretical results for linear AONTs and answer three open problems proposed in \cite{Stinson2017}. The first one about standard form is reported in Section \ref{problem}. Then we observe a general construction of linear $(2,p,p)$-AONT for all primes $p$ and give a positive answer to the existence results which attain the theoretical upper bounds of Theorem \ref{upper}, see Section \ref{symmetric}. Finally, we construct an infinite class of AONTs which is better than the AONTs defined by Cauchy matrices over $\mathbb{F}_q$, where $q$ is a prime power.

\subsection{Linear $(2,q,q)$-AONT in \emph{type $q-1$ standard form}}\label{problem}
First we define a `standard form' for a linear $(2,s,q)$-AONT.
\begin{definition}
Suppose $A$ is a matrix for a linear $(2,s,q)$-AONT. Then we can permute the rows and columns so that the $0$'s comprise the first $\mu$ entries on the main diagonal of $A$. If $\mu=0$, then we can multiply rows and columns by nonzero field elements so that all the entries in the first rows and first columns consist of $1$'s. If $\mu\ne 0$, we can multiply rows and columns by nonzero field elements so that all the entries in the first row and first column consist of $1$'s, except for the entry in the top left corner, which is a $0$. Such a matrix $A$ is said to be of \emph{type $\mu$ standard form}.
\end{definition}

In \cite{Stinson2017}, the authors obtained the structural conditions for a linear $(2,q,q)$-AONT and suggested an open problem as follows.
\begin{lemma}[{\rm \cite[Lemma 16]{Stinson2017}}]
Suppose $A$ is a matrix for a linear $(2,q,q)$-AONT in standard form. Then $A$ is of \emph{type $q$} or \emph{type $q-1$}.
\end{lemma}

{\bf \cite[Open Problem 4]{Stinson2017}:} As mentioned in Section 2.2, we performed exhaustive searches for linear $(2,q,q)$-AONT in type $q-1$ standard form, for all primes and prime powers $q \leq 9$, and found that no such AONT exists. We ask if there exists any linear $(2,q,q)$-AONT in type $q-1$ standard form.

In this subsection we give a negative answer to this open problem.
\begin{theorem}
For any prime power $q$, there does not exist a linear $(2,q,q)$-AONT in type $q-1$ standard form.
\end{theorem}

\proof Suppose, on the contrary, that $M$ is a matrix for a linear $(2,q,q)$-AONT in type $q-1$ standard form.  Then
 \[ M=\left(
    \begin{array}{cccccc}
      0 &1 & 1 & \cdots & 1 & 1\smallskip \\
      1 &0 & m_{2,3} & \cdots & m_{2,q-1} & m_{2,q} \smallskip \\
      1 &m_{3,2} & 0 & \cdots & m_{3,q-1}& m_{3,q}  \smallskip\\
     \vdots&\vdots &  \vdots & \ddots& \vdots & \vdots \smallskip \\
      1& m_{q-1,2} & m_{q-1,3} & \cdots & 0 & m_{q-1,q}\smallskip \\
      1& m_{q,2} & m_{q,3} & \cdots & m_{q,q-1} & m_{q,q}\smallskip \\
    \end{array}
  \right),\]
  where $m_{q,q}\neq 0$.  Since every 2 by 2 submatrix of $M$ is invertible, $m_{i,j}\neq 0$ for any $i\neq j$. Let
   \[ {M_1}=\left(
    \begin{array}{cccccc}
      0 &1/m_{q,2} & 1/m_{q,3} & \cdots & 1/m_{q,q-1} & 1/m_{q,q}\smallskip \\
      1 &0 & m_{2,3}/m_{q,3} & \cdots & m_{2,q-1}/m_{q,q-1} & m_{2,q}/m_{q,q} \smallskip \\
      1 &m_{3,2}/m_{q,2} & 0 & \cdots & m_{3,q-1}/m_{q,q-1}& m_{3,q}/m_{q,q}  \smallskip\\
     \vdots&\vdots &  \vdots & \ddots& \vdots & \vdots \smallskip \\
      1& m_{q-1,2}/m_{q,2} & m_{q-1,3}/m_{q,3} & \cdots & 0 & m_{q-1,q}/m_{q,q}\smallskip \\
      1& 1& 1 & \cdots & 1 & 1\smallskip \\
    \end{array}
  \right).\]
Clearly, $M_1$ is also invertible and every 2 by 2 submatrix of $M_1$ is invertible.
Consider 2 by 2 submatrices from the last row and and one row out of the first $q-1$ row. It is easy to see that
the first $q-1$ rows each have distinct entries. It follows that the sum of each row of $M_1$ is zero and the matrix $M_1$ is singular, a contradiction.\qed

\begin{corollary}
Suppose $A$ is a matrix for a linear $(2,q,q)$-AONT in standard form, then $A$ is of \emph{type $q$}.
\end{corollary}

\subsection{Construction for linear $(2,p,p)$-AONT for all primes $p$}\label{symmetric}

In \cite{Stinson2017}, the authors did an exhaustive search for a special subclass of linear $(2,p,p)$-AONTs and found that there exists a linear $(2,p,p)$-AONT for each value $p\in \{3,5,7,11,13,17,19,23,29\}$ and proposed the following question: are there infinitely many primes $p$ for which there exist linear $(2,p,p)$-AONT. We will give a construction to answer this problem.

\begin{construction}\label{inver}
Let $p$ be a prime and $A$ be a $p\times p$ matrix over $\f{p}$, where $A(s,t)$ denotes the entry in the $s$-th row and $t$-th column of $A$. Let $A(s,s)=0$ for $s=0,1,\ldots,p-1$, $A(s,1)=1$ for $s=1,2,\ldots,p-1$ and $A(s,t)=(s-t)^{-1}$ for $s=0,1,\ldots,p-1$, $t=1,2,\ldots,p-1$, $s\neq t$.
\end{construction}

\begin{example}
When $p=5$, the matrix $A$ would be
$$\left(
        \begin{array}{ccccc}
          0&4 & 2 & 3 & 1 \\
          1&0 & 4 & 2 & 3 \\
          1&1 & 0 & 4 & 2 \\
          1&3 & 1 & 0 & 4 \\
          1&2 & 3 & 1 & 0 \\
        \end{array}
      \right).$$
When $p=7$, the matrix $A$ would be
$$\left(
  \begin{array}{ccccccc}
    0&6 & 3 & 2 & 5 & 4 & 1 \\
    1&0 & 6 & 3 & 2 & 5 & 4 \\
    1&1 & 0 & 6 & 3 & 2 & 5 \\
    1&4 & 1 & 0 & 6 & 3 & 2 \\
    1&5 & 4 & 1 & 0 & 6 & 3 \\
    1&2 & 5 & 4 & 1 & 0 & 6 \\
    1&3 & 2 & 5 & 4 & 1 & 0 \\
  \end{array}
\right).$$
\end{example}

\begin{lemma}\label{twobytwo}
For any prime $p$, any $2\times 2$ submatrix of the matrix in Construction {\rm\ref{inver}} is invertible.
\end{lemma}

\begin{proof}
Consider a submatrix $A'$ defined by rows $i,j$ and columns $i',j'$, where $i<j$ and $i'<j'$. We consider the following cases:
\begin{enumerate}
  \item If $i=i'$ (or $i=j'$), then $\det(A')=-A(i,j')A(j,i')\ne 0$.
  \item If $i'=0$ and $i\ne 0$, then $\det(A')=A(j,j')-A(i,j')\ne 0$.
  \item If $i'\ne0$, $i\ne i',j'$ and $j\ne i',j'$, then $\det(A')=\frac{1}{i-i'}\frac{1}{j-j'}-\frac{1}{i-j'}\frac{1}{j-i'}$, so $\det(A')=0$ if and only if $ii'+jj'=ij'+i'j$. This condition is equivalent to $(i-j)(i'-j')=0$, which happens if and only if $i=j$ or $i'=j'$. We assumed that $i<j$ and $i'<j'$, so $A'$ is invertible.
\end{enumerate}
It is enough to show that any $2\times 2$ submatrix of the matrix in Construction \ref{inver} is invertible.\qed
\end{proof}

In order to show the construction above yields a linear $(2,p,p)$-AONT for $p>2$, it remains to show that $A$ is invertible. The trick of the proof is to introduce an auxiliary matrix which is closely related to the matrix we construct in Construction \ref{inver}.

\begin{construction}\label{auxiliary}
Let $p$ be a prime and $B$ be a $p\times p$ matrix over $\f{p}$, where $B(s,t)$ denotes the entry in the $s$-th row and $t$-th column of $B$. Let $B(s,s)=0$ for $s=0,1,\ldots,p-1$ and $B(s,t)=(s-t)^{-1}$ for $s=0,1,\ldots,p-1$, $t=0,1,\ldots,p-1$, $s\neq t$.
\end{construction}

\begin{remark}\label{rem}
Since the last $p-1$ columns of $A$ each contain every element of $\mathbb{F}_p$ exactly once, $A$ is invertible if and only if the lower right $p-1$ by $p-1$ submatrix of $A$ $($which is the same as the lower right $p-1$ by $p-1$ submatrix of $B)$ is invertible. Since each column and each row of $B$ contain every element of $\mathbb{F}_p$ exactly once, rank $(B)\leq p-1$, thereby, to prove the lower right $p-1$ by $p-1$ submatrix of $A$ is invertible  is equivalent to show the fact that rank($B$)$=p-1$.
\end{remark}

Next, we will regard the matrix $B$ as a generator matrix of  a cyclic code in coding theory. Before proceeding further, let us introduce some basic facts about cyclic codes.

A linear code $\cal C$ of length $n$ over $\mathbb{F}_q$ is cyclic provided that for each vector $\mathbf{c}=(c_0,c_1,\ldots,c_{n-1})$ in $\cal C$ the vector $(c_{n-1},c_0,\ldots,c_{n-2})$, obtained from $\mathbf{c}$ by the cyclic shift of coordinates $i\mapsto i+1\pmod n$, is also in $\cal C$. When examining cyclic codes over $\mathbb{F}_q$, it is convenient to represent the codewords in polynomial form. There is a bijective correspondence between the vectors $\mathbf{c}=(c_0,c_1,\ldots,c_{n-1})$ in $\mathbb{F}_q^n$ and the polynomials $c(x)=c_0+c_1x+\cdots+c_{n-1}x^{n-1}$ in $\mathbb{F}_q[x]$ of degree at most $n-1$. Notice that if $c(x)=c_0+c_1x+\cdots+c_{n-1}x^{n-1}$, then $xc(x)=c_{n-1}+c_0x+\cdots+c_{n-2}x^{n-1}$ if $x^n$ is set equal to $1$. Thus cyclic codes are ideals of $\mathbb{F}_q[x]/(x^n-1)$ and ideals of $\mathbb{F}_q[x]/(x^n-1)$ are cyclic codes. It is well known that there is a relation between the dimension of $\cal C$ and $\gcd(c(x),x^n-1)$, where $\cal C$ is generated by the polynomial $c(x)$ and the polynomial $\gcd(c(x),x^n-1)$ is called a generator polynomial.

\begin{theorem}[{\rm \cite[Theorem 4.2.1]{FECC}}]\label{dim}
Let $\cal C$ be a cyclic codes of length $n$ over $\mathbb{F}_q$ with a generator polynomial $g(x)$. Then the dimension of $\cal C$ is equal to $n-d$, where $d$ is the degree of $g(x)$.
\end{theorem}

\begin{lemma}\label{rank}
rank$(B)=p-1$.
\end{lemma}

\begin{proof}
First we observe that the matrix $B$ is cyclic and $f(x)=0-x-\frac{1}{2}x^2-\cdots-\frac{1}{p-1}x^{p-1}$. It is not difficult to verify that $f(1)=0$, $f'(1)\ne 0$  and $x^p-1=(x-1)^p$ over $\f{p}$, then $\gcd({f(x),x^p-1})=x-1$. By Theorem \ref{dim}, rank$(B)=p-1$ follows.
\end{proof}\qed

\begin{theorem}\label{main}
There exists a linear $(2,p,p)$-AONT for all primes $p$.
\end{theorem}

\begin{proof}
The theorem follows from Lemmas \ref{twobytwo} and \ref{rank}.
\end{proof}\qed

\subsection{Existence results for linear $(2,\Phi(q),q)$-AONT for prime power $q=p^r$}
In this subsection, we discuss linear $(2,s,q)$-AONTs for prime powers $q=p^r$. As a consequence, we improve the lower bound on $M(q)$ in general and answer an open problem proposed in \cite{Stinson2017}. To the best of our knowledge, there are only two systematic results in this topic. In \cite{Stinson2001}, Cauchy matrices were mentioned as a possible method of constructing AONT, see Theorem \ref{cauchy}. Recently, Nasr Esfahani et al. \cite{Stinson2017} pointed out that Vandermonde matrices can be treated as a method of constructing AONTs and get the following theorem.

\begin{theorem}[\rm{\cite[Theorem 2.1]{Stinson2017}}]\label{power}
Suppose $q=2^n$, $q-1$ is prime and $s\leq q-1$. Then there exists a linear $(2,s,q)$-AONT over $\f{q}$.
\end{theorem}

\begin{remark}
Since the parameter $s$ can attain $q-1$, by Theorem $\ref{upper}$, the result above yields a good construction for linear AONT. However, it actually requires that $2^n-1$ is a Mersenne prime, up to now, only $50$ such primes are known. Thus, the construction above cannot provide a infinite class of linear AONT.
\end{remark}

Next, we will provide a general construction of linear $(2,\Phi(q),q)$-AONT for all prime powers $q=p^r$. We divide our proof in two steps. First, we construct a $q-1$ by $q-1$ matrix  over $\f{q}$ such that any $2$ by $2$ submatrix is invertible. Then, we find a $\Phi(q)$ by $\Phi(q)$ submatrix which is invertible.

\begin{construction}\label{primepower}
Let $q=p^r$ be a prime power, $\alpha$ be a primitive element of $\f{q}$ and $P$ be a $q-1$ by $q-1$ matrix over $\f{q}$, where $P(s,t)$ denotes the entry in the $s$-th row and $t$-th column of $P$. Let $P(s,s)=0$ for $s=0,1,\ldots,q-2$ and $P(s,t)=\frac{\alpha^s}{\alpha^s-\alpha^t}$ for $s=0,1,\ldots,q-2$, $t=0,1,\ldots,q-2$, $s\neq t$.
\end{construction}

\begin{example}
When $q=8$, let $\alpha$ be a primitive element of $\f{8}$ defined by $\alpha^3+\alpha+1=0$, the matrix $P$ defined over $\f{8}$ would be
$$\left(
        \begin{array}{ccccccc}
          0&\alpha^4& \alpha & \alpha^6 & \alpha^2&\alpha^3&\alpha^5 \\
          \alpha^5&0 & \alpha^4 & \alpha & \alpha^6&\alpha^2&\alpha^3 \\
          \alpha^3&\alpha^5 & 0 & \alpha^4 & \alpha&\alpha^6&\alpha^2 \\
          \alpha^2&\alpha^3 & \alpha^5 & 0 & \alpha^4&\alpha&\alpha^6 \\
          \alpha^6&\alpha^2 & \alpha^3 & \alpha^5 & 0&\alpha^4&\alpha \\
          \alpha&\alpha^6&\alpha^2&\alpha^3&\alpha^5&0&\alpha^4\\
          \alpha^4&\alpha&\alpha^6&\alpha^2&\alpha^3&\alpha^5&0\\
        \end{array}
      \right).$$
When $q=9$, let $\beta$ be a primitive element of $\f{9}$ defined by $\beta^2+2\beta+2=0$, the matrix $P$ defined over $\f{9}$ would be
$$\left(
  \begin{array}{cccccccc}
    0&\beta^5 & \beta^3 & \beta^7 & 2 & \beta^6 & \beta&\beta^2 \\
    \beta^2&0&\beta^5 & \beta^3 & \beta^7 & 2 & \beta^6 & \beta\\
     \beta&\beta^2&0&\beta^5 & \beta^3 & \beta^7 & 2 & \beta^6  \\
     \beta^6&\beta&\beta^2&0&\beta^5 & \beta^3 & \beta^7 & 2   \\
    2 &\beta^6&\beta&\beta^2&0&\beta^5 & \beta^3 & \beta^7  \\
     \beta^7&2 &\beta^6&\beta&\beta^2&0&\beta^5 & \beta^3  \\
     \beta^3&\beta^7&2 &\beta^6&\beta&\beta^2&0&\beta^5  \\
     \beta^5 &\beta^3&\beta^7&2 &\beta^6&\beta&\beta^2&0 \\
  \end{array}
\right).$$
\end{example}

\begin{lemma}\label{two2}
For any prime power $q$, any $2\times 2$ submatrix of the matrix in Construction \ref{primepower} is invertible.
\end{lemma}

\begin{proof}
Consider a submatrix $P'$ defined by rows $i,j$ and columns $i',j'$, where $i<j$ and $i'<j'$. We consider the following cases:
\begin{enumerate}
  \item If $i=i'$ (or $i=j'$, $j=i'$, $j=j'$), then $\det(P')=-P(i,j')P(j,i')\ne 0$.
  \item Otherwise, $\det(P')=\frac{\alpha^i}{\alpha^i-\alpha^{i'}}\frac{\alpha^j}{\alpha^j-\alpha^{j'}}-\frac{\alpha^i}{\alpha^i-\alpha^{j'}}\frac{\alpha^j}{\alpha^j-\alpha^{i'}}$, so $\det(P')=0$ if and only if $(\alpha^i-\alpha^j)(\alpha^{i'}-\alpha^{j})=0$. This condition happens if and only if $i=j$ or $i'=j'$. We assumed that $i<j$ and $i'<j'$, so $P'$ is invertible.
\end{enumerate}
It is enough to show that any $2\times 2$ submatrix of the matrix in Construction \ref{primepower} is invertible.\qed
\end{proof}

Next, we will prove that the rank of $P$ is $\Phi(q)$.

\begin{lemma}\label{rank2}
rank$(P)=\Phi(q)$.
\end{lemma}

\begin{proof}
We observe that the matrix $P$ is cyclic and $f(x)=0+\frac{1}{1-\alpha}x+\frac{1}{1-\alpha^2}x^2+\cdots+\frac{1}{1-\alpha^{q-2}}x^{q-2}$. Since $x^{q-1}-1=(x-1)(x-\alpha)\cdots(x-\alpha^{q-2})$ in $\f{q}$, by Theorem \ref{dim}, in order to prove this lemma, we need to show $f(x)$ has exactly $p^{r-1}-1$ distinct roots in $\f{q}$.
\begin{claim}\label{claim}
$$\{x\in\f{q}\mid f(x)=0\}=\{\alpha^p,\alpha^{2p},\ldots,\alpha^{(p^{r-1}-1)p}\}.$$
\end{claim}
We consider the following three cases:
\begin{enumerate}
  \item If $x=1$, then $f(x)=\frac{1}{1-\alpha}+\frac{1}{1-\alpha^2}+\cdots+\frac{1}{1-\alpha^{q-2}}=-1\ne 0$.
  \item If $x=\alpha^{kp}$, where $1\leq k\leq p^{r-1}-1$, then $-f(\alpha^{kp})=1+\frac{1}{1-\alpha}(1-x)+\frac{1}{1-\alpha^2}(1-x^2)+\cdots+\frac{1}{1-\alpha^{q-2}}(1-x^{q-2})$. We expand  $-f(\alpha^{kp})$ as follows:
      $$\frac{1}{1-\alpha^i}(1-(\alpha^{kp})^i)=\frac{1}{1-\alpha^i}(1-(\alpha^{i})^{kp})=1+\alpha^i+(\alpha^i)^2+\cdots+(\alpha^i)^{kp-1}.$$
      Then $-f(\alpha^{kp})=q-1+\sum_{j=1}^{kp-1}\sum_{i=1}^{q-2}(\alpha^i)^j=q-1+(kp-1)\cdot(-1)=0$. Thus $\alpha^{kp}$, $1\leq k\leq p^{r-1}-1$, are the roots of $f(x)$.
  \item If $x=\alpha^{kp+r}$, where $1\leq r<p$ and $1\leq k\leq p^{r-1}-1$, then $-f(\alpha^{kp+r})=q-1+\sum_{j=1}^{kp+r-1}\sum_{i=1}^{q-2}(\alpha^i)^j=q-1+(kp+r-1)\cdot(-1)=-r\ne0$.
\end{enumerate}
This completes the proof of Claim \ref{claim}, and Lemma \ref{rank2} follows.
\end{proof}\qed

\begin{theorem}\label{main2}
There exists a linear $(2,\Phi(q),q)$-AONT for all primes power $q$.
\end{theorem}

\begin{proof}
The theorem follows from Lemmas \ref{two2} and \ref{rank2}.
\end{proof}\qed

In \cite{Stinson2017}, the authors did an exhaustive search and observed a very interesting structure for linear $(2,q,q)$-AONTs in type $q$ standard form, which the authors call it `$\tau$-skew-symmetric'.  Suppose $M$ is a matrix for a linear $(2,q,q)$-AONT in type $q$ standard
form. We say that the matrix $M$ is $\tau$-skew-symmetric if for any pair of cells $(i,j)$ and $(j,i)$ of $M$, where $2\leq i,j\leq q$ and $i\ne j$, it holds that $M(i,j)+M(j,i)=\tau$. Furthermore, we say that $M$ is cyclic if $M_1$ (the lower right $q-1$ by $q-1$ submatrix of $M$) is a cyclic matrix. Up to equivalence, the authors in \cite{Stinson2017} actually found there was exactly one $(q-1)$-skew-symmetric $(2,q,q)$-AONT for each value $q\in\{3,5,7,11,13,17,19,23,29\}$ and proposed the following open problem:

{\bf \cite[Open Problem 2]{Stinson2017}:} Are there infinitely many primes $p$ for which there exist (cyclic) skew-symmetric $(2,p,p)$-AONT?

For the matrix $P$ in Construction \ref{primepower}, when $q=p$ is a prime, we construct a linear $(2,p-1,p)$-AONT. Let $Q$ be a $p$ by $p$ matrix in type $p$ standard form and choose the matrix $P$ as $Q_1$. Then the matrix $Q$ is as following.

\begin{construction}
\label{1-skew}
Let $p$ be a prime and $\alpha$ be a primitive element of $\f{p}$. Define a $p$ by $p$ matrix $Q$ over $\f{p}$ as follows, where $Q(s,t)$ denotes the entry in the $s$-th row and $t$-th column of $Q$. Let $Q(s,s)=0$ for $s=0,1,\ldots,p-1$, $Q(s,0)=1$ and $Q(0,s)=1$ for $s=1,2,\ldots,p-1$ and $Q(s,t)=\frac{\alpha^t}{\alpha^t-\alpha^s}$ for $s=1,2,\ldots,p-1$, $t=1,2,\ldots,p-1$ and $s\neq t$.
\end{construction}

By the property of the matrix $P$ in Construction \ref{primepower}, it is easy to see that $Q_1$ is cyclic and $1$-skew-symmetric ($Q(s,t)+Q(t,s)=1$ for $1\leq s,t\leq p-1$, $s\ne t$). By the same argument of Remark \ref{rem},  the matrix $Q$ is invertible and every $2\times 2$ submatrix of the matrix $Q$ is also invertible. For $\tau \in \mathbb{F}_p\setminus \{0\}$, when we replace the lower right $p-1$ by $p-1$ submatrix $Q_1$ of $Q$ with $\tau Q_1$, the resultant $p$ by $p$ matrix is $\tau$-skew-symmetric. The corollary is a direct consequence of what we observed, which answers the open problem mentioned above.

\begin{corollary}
There exists a cyclic skew-symmetric $(2,p,p)$-AONT for any prime $p$.
\end{corollary}

We should note that Construction \ref{1-skew} gives a theoretical explanation of the exhaustive searches in \cite[Section \uppercase\expandafter{\romannumeral3}.B]{Stinson2017}. When $p=7$, $5$ is a primitive element of $\f{7}$. Suppose $\tau=6$, we choose the lower right $6$ by $6$ submatrix as $6Q_1$ in Construction \ref{1-skew}, then the matrix $Q$ would be

$$\left(
  \begin{array}{ccccccc}
    0&1 & 1 & 1 & 1 & 1 & 1 \\
    1&0&2 & 5 & 3 & 1 & 4 \\
     1&4&0&2 & 5 & 3 & 1   \\
     1&1&4&0&2 & 5 & 3   \\
    1 &3&1&4&0&2 & 5   \\
     1&5 &3&1&4&0&2   \\
     1&2&5 &3&1&4&0  \\
  \end{array}
\right),$$
which is coincident with \cite[Example 35]{Stinson2017}.

\subsection{Dual property of linear AONTs}
In this subsection, we will show the dual property of linear $(t,s,q)$-AONT.

Let $A$ be an $N$ by $k$ array whose entries are elements chosen from an alphabet $X$ of size $v$. We will refer to $A$ as an $(N,k,v)$-array. Suppose the columns of $A$ are labelled by the elements in the set $C = \{1,\ldots,k\}$. Let $D \subset C$ and define $A_D$ to be the subarray obtained from $A$ by deleting all the columns $c \not\in D$. We say that $A$ is unbiased with respect to $D$ if the rows of $A_D$ contain every $|D|$-tuple of elements of $X$ exactly $N/v^{|D|}$ times.
The following result characterizes $(t,s,v)$-AONT in terms of arrays that are unbiased with respect to certain subsets of columns.

\begin{theorem}\label{OA}
{\rm (\cite[Theorem 34]{DES2016})} A $(t, s, v)$-AONT is equivalent to a $(v^s, 2s, v)$-array that is unbiased with respect to the following subsets of columns:

$1.$ $\{1,\ldots,s\}$,

$2.$ $\{s+1,\ldots,2s\}$, and

$3.$ $I\cup \{s+1,\ldots,2s\}\setminus J$, for all $I\subset \{1,\ldots,s\}$ with $| I |=t$ and all $J\subset \{s+1,\ldots,2s\}$ with $| J | = t$.

\end{theorem}

It is easy to see that if a $(v^s, 2s, v)$-array $A$ satisfies the conditions in Theorem \ref{OA}, then the array $\widetilde{A}$ obtained by interchanging the first $s$ columns and the last $s$ columns of $A$ is unbiased with respect to the following subsets of columns:

$1.$ $\{1,\ldots,s\}$,

$2.$ $\{s+1,\ldots,2s\}$, and

$3.$ $I\cup \{s+1,\ldots,2s\}\setminus J$, for all $I\subset \{1,\ldots,s\}$ with $| I |=s-t$ and all $J\subset \{s+1,\ldots,2s\}$ with $| J | = s-t$.

\noindent Therefore, $\widetilde{A}$ is an $(s-t, s, v)$-AONT.

\begin{theorem}
\label{dual-nonlinear}
If there exists a $(t, s, v)$-AONT, then there exists an $(s-t, s, v)$-AONT.
\end{theorem}

By the definition of linear AONT, i.e., the relationship between the first $s$ columns and the last $s$ columns as in (\ref{def-LAONT}), interchanging the first $s$ columns and the last $s$ columns of a linear $(t, s, q)$-AONT yields a linear $(t-s, s, q)$-AONT by Theorem \ref{dual-nonlinear}. So, the matrix $M^{-1}$ must satisfy that every $t-s$ by $t-s$ submatrix is invertible.

\begin{theorem}
\label{dual-LinearAONT}
If there exists a linear $(t, s, q)$-AONT with the corresponding matrix $M$, then the linear AONT defined by $M^{-1}$ is a linear $(s-t, t, q)$-AONT.
\end{theorem}

By Theorem \ref{upper} and \ref{dual-LinearAONT}, we have the following.

\begin{corollary}
There is no linear $(q-1,q+1,q)$-AONT for any prime power $q>2$.
\end{corollary}

By Theorem \ref{dual-LinearAONT} and the constructions in Theorems \ref{main}, \ref{main2} and  \cite[Theorem 11]{Stinson2017}, we have the following existence results.

\begin{corollary}
Suppose $p$ is a prime. Then there exists a linear $(p-2,p,p)$-AONT over $\mathbb{F}_p$.
\end{corollary}

\begin{corollary}
Suppose $q$ is a prime power. Then there exists a linear $(\Phi(q)-2,\Phi(q),q)$-AONT over $\mathbb{F}_q$.
\end{corollary}
\begin{corollary}
Suppose $q=2^n$, $q-1$ is a prime. Then there exists a linear $(q-3,q-1,q)$-AONT over $\mathbb{F}_q$.
\end{corollary}

\begin{remark}
The above result requires that $2^n-1$ is a (Mersenne) prime. Here are a couple of results on Mersenne primes from {\rm \cite{MersennePrime}}. The first few Mersenne primes occur for
$$n = 2,3,5,7,13,31,61,89,107,127.$$
As Nasr Esfahani et al. pointed out {\rm \cite{Stinson2017}} that there were $50$ known Mersenne primes, the largest being
$2^{77232917} -1$, which was discovered in January $2018$.
\end{remark}

\section{General AONT}\label{3}
In this section, we present a recursive construction for general AONTs and a new relationship between AONTs and orthogonal arrays. Based on the observations above, we get a general construction for nonlinear $(2,3,n)$-AONT, except for $n=2,6$.

We begin with a direct product construction.
\begin{lemma}[{\rm{Product Construction}}]\label{productconstruction}
If there is a $(t,s,n)$-AONT and a $(t,s,m)$-AONT, then there exists a $(t,s,mn)$-AONT.
\end{lemma}

\begin{proof}
Let $A=[(a_{i,j})]$ be an $(n^s,2s,n)$-array over $\mathbb{Z}_n$ corresponding to a $(t,s,n)$-AONT, and $B=[(b_{k,j})]$ be an $(m^s,2s,m)$-array over $\mathbb{Z}_m$ corresponding to a $(t,s,m)$-AONT.
For $1\leq i\leq n^s$ and $1\leq j\leq m^s$, denote
\[
H_{\{i,j\}}=\left ( (a_{i,1},b_{j,1}),(a_{i,2},b_{j,2}),\ldots,(a_{i,2s},b_{j,2s})\right ).
\]
It is routine to check that the array consisting of all row vectors $H_{\{i,j\}}$ where $1\leq i\leq n^s$ and $1\leq j\leq m^s$,
is a $((nm)^s,2s,nm)$-array over $\mathbb{Z}_n\times\mathbb{Z}_m$ corresponding to a $(t,s,nm)$-AONT. \qed
\end{proof}

To obtain a relation between AONTs and orthogonal arrays, Stinson \cite{Stinson2001} completely determined the existence of $(1,2,v)$-AONTs.

\begin{theorem}{\rm \cite[Theorem 3.4]{Stinson2001}}
There exists a $(1,2,v)$-AONT if and only if $v\ne 2,6$.
\end{theorem}

For prime powers $q$, the existence of $(1, s, q)$-AONT has been completely determined in \cite{Stinson2001}.

\begin{theorem}
{\rm \cite[Corollary 2.3]{Stinson2001}}
\label{linear(1, s, q)-AONT} There exists a linear $(1, s, q)$-AONT for all prime powers
$q > 2$ and for all positive integers $s$.
\end{theorem}

Applying Lemma \ref{productconstruction} with the known linear $(1, s, q)$-AONT in Theorem \ref{linear(1, s, q)-AONT}, we could generalise the existence results of $(1,s,v)$-AONT for $s\geq 3$.

\begin{corollary}
There exists a  $(1, s, n)$-AONT for any integer $n\geq 3$, $n\not\equiv 2\pmod 4$ and for all positive integers $s$.
\end{corollary}

Next we will obtain a new relationship between AONTs and orthogonal arrays. A $(v^t, k, v)$-array over $X$ is called an orthogonal array and denoted by OA$(t,k,v)$ if it is unbiased with respect to every $t$-subset of columns $\{1,2,\ldots,k\}$.
The following relationship between OA and AONT is immediate from Theorem \ref{OA}.

\begin{corollary}[\rm{\cite[Corollary 35]{DES2016}}]
If there exists an OA$(s,2s,v)$, then there exists a $(t,s,v)$-AONT for all $t$ such that $1\leq t\leq s$.
\end{corollary}

\begin{theorem}[\rm{\cite[Theorem 23]{Stinson2017}}]
Suppose there is a $(t,s,v)$-AONT. Then there is an OA$(t,s,v)$.
\end{theorem}

Roughly speaking, a $(t,s,v)$-AONT is a combinatorial configuration between OA$(t,s,v)$ and OA$(s,2s,v)$. The main difference of our results is that we do not construct an AONT directly from the structure of OA, actually, we take an OA as the auxiliary array in our construction.
\begin{theorem}\label{OA(4,v)-(2,3,v)-AONT}
If there exists an OA$(2,4,v)$, then there exists a $(2,3,v)$-AONT.
\end{theorem}

\begin{proof}
Let $A$ be an OA$(2,4,v)$ over $\mathbb{Z}_v$, whose rows are indexed by $C_{i,j}$, $i,j\in \mathbb{Z}_v$. Without loss of generality, we assume that $C_{i,j}=(i,j,L_1(i,j),L_2(i,j))$.
For $i,j,x\in \mathbb{Z}_v$, we construct $v^3$ row vectors as follows:
\[
H_{i,j,x}=(L_1(i,j),L_2(i,j),x,j+x,L_1(i,x),L_2(i,x)).
\]
Let $H$ be a $(v^3,6,v)$-array consisting of the $v^3$ row vectors above. We claim that $H$ is a $(2,3,v)$-AONT. To prove this claim, we check the conditions $(1)$, $(2)$ and $(3)$ of Theorem \ref{OA}.

Case 1. The first two conditions follow immediately from the definition of OA.

Case 2. Suppose we choose two columns $\{c_1,c_2\}$ from $\{1,2,3\}$ and one column $\{c_3\}$ from $\{4,5,6\}$, it suffices to establish the bijection between the row vectors chosen above and all $3$-tuples. By the symmetry of our construction, without loss of generality, we assume that $c_1=1,c_2=2$ and $c_3=4$. For any $3$-tuple $(a,b,c)\in \mathbb{Z}_v^3$, let $L_1(i,j)=a$, $L_2(i,j)=b$ and $j+x=c$. By the definition of OA, we could uniquely determine the index $i,j$ from the first two equalities, then uniquely determine the index $x$ from the last one.
The proof is completed.  \qed
\end{proof}

It is well known that there is an OA$(2,4,n)$ if and only if $n\neq 2,6$ \cite{BSP1960}. Applying Theorem \ref{OA(4,v)-(2,3,v)-AONT} yields the following.

\begin{corollary}
For any positive integer $n$ with $n\neq 2,6$, there is a $(2,3,n)$-AONT and a $(1,3,n)$-AONT.
\end{corollary}

\section{Conclusions}\label{4}
In this paper, we continue the study of t-all-or-nothing transforms over alphabets of arbitrary size. First we solved three open problems proposed in \cite{Stinson2017} and showed that there exists linear $(2,p,p)$-AONT. Then  for prime powers $q$, we construct the first infinite class of linear AONTs over $\mathbb{F}_q$ which is better than the linear AONTs defined by Cauchy matrices. Besides, we also presented a recursive construction for general AONTs and a new relationship between AONTs and orthogonal arrays.

\end{document}